\documentclass{amsart}
\newtheorem{theorem}{Theorem}[section]
\newtheorem{proposition}[theorem]{Proposition}

\newtheorem{remark}[theorem]{Remark}
\def\pagenumber{1}

\usepackage{comment}

\begin{document}
\setcounter{page}{\pagenumber}
\newcommand{\T}{\mathbb{T}}
\newcommand{\R}{\mathbb{R}}
\newcommand{\Q}{\mathbb{Q}}
\newcommand{\N}{\mathbb{N}}
\newcommand{\Z}{\mathbb{Z}}
\newcommand{\tx}[1]{\quad\mbox{#1}\quad}

\title[Hamiltonian potentials and quartic NVE]{\mbox{}\\[1cm]On Hamiltonian potentials with quartic polynomial normal variational equations}
\author[P. Acosta-Hum\'anez]{Primitivo B. Acosta-Hum\'anez$^1$}
\author[D. Blazquez-Sanz]{David Blazquez-Sanz$^2$}
\author[C. Vargas Contreras]{Camilo A. Vargas-Contreras $^2$}
\maketitle
\vspace{-.5cm}

{\footnotesize
\begin{center}
$^1$Departament Matem\`atica Aplicada I, Universitat Polit\`ecnica
de Catalunya, Barcelona, Spain.

$^2$Escuela de Matem\'a ticas, Universidad Sergio Arboleda,
Bogot\'a, Colombia.
\end{center}
\vspace{.5cm}

\noindent{\bf ABSTRACT.} In this paper we prove that there exists
only one family of classical Hamiltonian systems of two degrees of
freedom with invariant plane $\Gamma=\{q_2=p_2=0\}$ whose normal
variational equation around integral curves in $\Gamma$ is
generically a Hill-Schr\"odinger equation with quartic polynomial
potential. In particular, by means of the Morales-Ramis theory,
these Hamiltonian systems are non-integrable through rational
first integrals.

\bigskip


\bigskip

 }
\section{\bf INTRODUCTION}\label{s1}
The non-integrability in the Liouville sense of some physical
problems, in particular Hamiltonian systems of two degrees of
freedom, such as Henon-Heiles systems, three body problems,
Bianchi models, etc., has been
 studied by many authors (see \cite{Audin,MorMonograph} and references therein).
 In \cite{Audin} it was shown that the Henon-Heiles system given by
\begin{displaymath} H(x_1,x_2,y_1,y_2)={1\over
2}(y_1^2+y_2^2)-x_2^2(A+x_1)-{\lambda\over3}x_1^3
\end{displaymath}
does not admit an additional rational first integral for
$\lambda=0$. The approach used there was by means of the
Morales-Ramis theory, analyzing the differential Galois groups of
the normal variational equations (NVEs) around an invariant plane
$\Gamma=\{x_2=y_2=0\}$. According to the Morales-Ramis theory, if
the differential Galois group of one NVE is a non-virtually
abelian group, i.e. the identity connected component is a
non-abelian group, then the Hamiltonian system is non-integrable
in the Liouville sense (see \cite{MorRamis1, MorRamis2} and also
\cite{MorMonograph}).

One result obtained in \cite{acbl} and presented later in
\cite{acbl2}, is the following theorem.\\

{\bf Theorem.} \emph{The Galois group of the Hill-Schr\"odinger
equation,
$$\ddot \xi = P_n(t)\xi,$$
where $P_n(t)\in\mathbb C[t]$ is a non-constant polynomial, is a
non-abelian connected group isomorphic either to
$\mathrm{SL}(2,\mathbb C)$ or to the semidirect product of
$\mathbb{C}^*$ with $\mathbb{C}$ (also known as the Borel group).}
\\

 Now, assuming that $\beta(x_1,x_2)$ is analytic around $\Gamma$ and considering the following
\textit{generalization} of the Hen\'on-Heiles system
\begin{equation}\label{genhh} H(x_1,x_2,y_1,y_2)={y_1^2+y_2^2\over 2}-x_2^2(A_0+A_1x_1+\ldots+A_nx_1^n)-{\lambda\over3}x_1^3+\beta(x_1,x_2)x_2^3,
\end{equation}
it is proven in \cite{acbl} that \eqref{genhh} does not admit an
additional rational first integral for $\lambda=0$ because the
NVEs around $\Gamma$ are given by $\ddot \xi=P_n(t)\xi$.

The original problem considered in \cite{acbl}, see also
\cite{acbl2}, was given as follows:

\emph{\begin{center} Given the differential equation $\ddot
x=r(t)x$, find all the families of\\ Hamiltonian systems with
two degrees of freedom in which there exists an invariant plane
$\Gamma=\{x_2=y_2=0\}$ and the normal variational equations\\
around this plane are given by $\ddot x=r(t)x$. In particular,
what happens with $r(t)=P_n(t)$ and with $r(t)=a\cos\omega
t+b\sin\omega t$?
\end{center}
}
\medskip

In the polynomial case, it was proven that, if $n$ is odd, then
there exists only one family of potentials with this polynomial
NVE, which up to constants corresponds to the potential of the
generalized Hen\'on-Heiles system \eqref{genhh}. However, for the
quadratic case and a particular case of the classical Mathieu equation (obtained by a scalling and a shift in time) we discovered
several families of potentials falling in the same NVE. The
polynomial case with $n>2$, $n$ even, was set as an open question.
In this paper we prove that for $n=4$, the so-called {\it quartic
case}, there exists only one family of potentials with this
polynomial NVE, which up to constants corresponds to the potential
of the generalized Hen\'on-Heiles system \eqref{genhh} for
$\lambda=0$.
\\

Related problems with this approach have been studied before by
Morales and Sim\'o (see \cite{mosi} ) and by Baider, Churchill and
Rod (see \cite{bachro}). The use of techniques of Differential
Galois theory to determine the non-integrability of Hamiltonian
systems appeared independently for first time in \cite{MorThs,
MorSimo1} and \cite{ChRo1991}, followed by \cite{bachro},
\cite{ChRoSi1995} and \cite{mosi}. A common limitation presented
in these works is that they only analyzed cases of fuchsian
monodromy groups, avoiding cases of irregular singularities of
linear differential equations. The case of the NVEs with irregular
singularities can be approached from the Morales-Ramis framework
(\cite{MorRamis1, MorRamis2}, see also \cite{MorMonograph}).

\section{Morales-Ramis theory}
In this section we set the theoretical background needed to
understand the rest of the paper.

\subsection{Integrability of Hamiltonian systems}
  A symplectic manifold (real or complex) $M_{2n}$ is a $2n$-dimensional manifold,
provided with a non-degenerate closed 2-form $\omega_2$. This
closed $2$-form, the so-called \emph{symplectic form},
 gives us a natural isomorphism
between vector bundles, $\flat\colon TM\to T^*M$. Given a function
$H$ on $M$, there is an unique vector field $X_H$ such that
  $$\flat(X_H) = dH,$$
which is \emph{the Hamiltonian vector field of $H$}. Furthermore,
it has the structure of a \emph{Poisson algebra} over the ring of
differentiable functions of $M_{2n}$ by defining:
$$\{H,F\} := X_H F.$$
  We say that $H$ and $F$ are \emph{in involution} if and only if $\{H,F\} = 0$.
From our definition, it is obvious that $F$ is a \emph{first
integral} of $X_H$ if and only if $H$ and $F$ are in involution.
In particular $H$ is always a first integral of $X_H$. Moreover,
if $H$ and $F$ are in involution, then their flows commute.

In a system of canonical coordinates, $p_1,\ldots,p_n$,
$q_1,\ldots,q_n$ the symplectic form $\omega_2$ is given by
$\omega_2=\sum_{1=1}^n dp_i \wedge dq_i$ and the equations of the
flow of $X_H$ can be written in the form
\begin{equation}\label{hameqs}
\dot q = \frac{\partial H}{\partial p} \left( = \{H, q\}\right),
\quad \dot p  = - \frac{\partial H}{\partial q} \left(= \{ H,
p\}\right).
\end{equation}
The equations given in expression \eqref{hameqs} are known as
\emph{Hamilton equations}.

\begin{theorem}[Liouville-Arnold]
  Let $X_H$ be a Hamiltonian defined on a real symplectic manifold
$M_{2n}$. Assume that there are $n$ functionally independent first
integrals $F_1,\ldots, F_n$ in involution.
  Let $M_a$ be a non-singular (that is, $dF_1,\ldots, dF_n$ are independent  over each point of
$M_a$) level manifold, $$M_a = \{p\colon F_1(p)=a_1,\ldots,
F_n(p)=a_n\}.$$
\begin{enumerate}
\item If $M_a$ is compact and connected, then it is a torus
$M_a\simeq \mathbb R^n/\mathbb Z^n$.
\item In a neighborhood of the torus $M_a$ there are functions
$I_1,\ldots I_n,\phi_1,\ldots,\phi_n$ such that
$$\omega_2 = \sum_{i=1}^n d I_i\wedge d\phi_i,$$
and $\{H,I_j\} = 0$ for $j = 1,\ldots, n$.
\end{enumerate}
\end{theorem}

From now on, we will consider $\mathbb{C}^{2n}$ as a complex
symplectic manifold. The Liouville-Arnold theorem gives us a
notion of integrability for Hamiltonian systems. A Hamiltonian $H$
in $\mathbb{C}^{2n}$ is called \emph{integrable in Liouville's
sense} if and only if there exist $n$ independent first integrals
of $X_H$ in involution. We will say that $H$ is integrable
\emph{by rational functions} if and only we can find a complete
set of first integrals within the family of rational functions.

\subsection{Variational equations}

We want to relate the integrability of Hamiltonian systems with
the Picard-Vessiot theory. We deal with non-linear Hamiltonian
systems. Nevertheless, given a Hamiltonian $H$ in
$\mathbb{C}^{2n}$ and $\Gamma$ an integral curve of $X_H$, we can
consider the \emph{first variational equation} (VE), namely
$$\mathcal{L}_{X_H}\xi=0,$$
in which the linear equation is induced over the tangent bundle
($\xi$ represents a vector field supported on $\Gamma$).

 Let $\Gamma$ be parameterized by $\gamma\colon t\mapsto (x(t),y(t))$
in such way that
$$\frac{d x_i}{d t} = \frac{\partial H}{\partial y_i}, \quad
\frac{d y_i}{dt} = - \frac{\partial H}{\partial x_i}.$$

  Then the VE along $\Gamma$ is the linear system,
$$\left(\begin{array}{c} \dot \xi_i \\ \dot\eta_i, \end{array}\right) =
\left(\begin{array}{cc}
\frac{\partial^2 H}{\partial y_i\partial x_j}(\gamma(t)) & \frac{\partial^2 H}{\partial y_i \partial y_j}(\gamma(t)) \\
- \frac{\partial^2 H}{\partial x_i \partial x_j}(\gamma(t)) &
-\frac{\partial^2 H}{\partial x_i \partial y_j}(\gamma(t))
\end{array}\right)
 \left(\begin{array}{c}\xi_i
\\ \eta_i \end{array}\right).$$

  From the definition of Lie derivative, it follows that
$$\xi_i(t) = \frac{\partial H}{\partial y_i} (\gamma(t)),\quad
\eta_i(t) = - \frac{\partial H}{\partial x_i}(\gamma(t)),$$ is a
solution of the VE. We can use a generalization of D'Alambert's
method to reduce our VE (see \cite{MorRamis1, MorRamis2} and see
also \cite{MorMonograph}), obtaining the so-called \emph{normal
variational equation} (the NVE). We can see that the NVE is a
linear system of rank $2(n-1)$. In the case of Hamiltonian systems
of two degrees of freedom, their NVE can be seen as second order
linear homogeneous differential equation.

\subsection{Non-integrability tools}

The Morales-Ramis theory relates the integrability of Hamiltonian
systems in the Liouville sense with the integrability of linear
differential equations  in the sense of differential Galois theory
(see \cite{MorRamis1, MorRamis2} and see also
\cite{MorMonograph}). In such an approach the linearization
(variational equations) of Hamiltonian systems along some known
particular solution is studied. If the Hamiltonian system is
integrable in the Liouville sense, then we expect that the
linearized equation has good properties in the sense of
differential Galois theory (also known as Picard-Vessiot theory).
To be more precise, for integrable Hamiltonian systems, the Galois
group of the linearized equation must be virtually abelian, i.e.
its identity connected component is abelian. This gives us the
best non-integrability criterion known so far for Hamiltonian
systems. This approach has been extended to higher order
variational equations in \cite{MorRamSimo} and also to
non-autonomous Hamiltonian systems in \cite{ac}.
\\

The Morales-Ramis theory is composed by several results relating
the existence of first integrals of $H$ with the Galois group of
the variational equations (see for example \cite{MorRamis1},
\cite{MorRamis2} and see also \cite{MorMonograph}).

\medskip

Most applications of the Picard-Vessiot theory to the
integrability analysis, are studied considering meromorphic
functions, due to the fact that the NVE are of hypergeometric type
(every singular point is a singular regular point, including the
points at infinity). In the case of polynomial NVE there exists
only one singular point, the point at infinity, $t=\infty$, which
is an irregular singular point. Hence we will only work with
particular solutions in the context of meromorphic functions with
certain properties of regularity near to the infinity point, that
is, rational functions of the \emph{positions and momenta}. In
this context, the Galoisian obstruction is given by means of
rational functions (see for example \cite{MorRamis1},
\cite{MorRamis2} and see also \cite{MorMonograph}).

In this paper we will use the following result:

\begin{theorem}[\cite{MorRamis1}]\label{:MR}
  Let $H$ be a Hamiltonian in $\mathbb{C}^{2n}$ and $\gamma$ a particular
solution such that the {\rm NVE} has irregular singularities at
points of $\gamma$ at infinity. Then, if $H$ is completely
integrable by rational functions, then the identity component of
the Galois Group of the  {\rm NVE} is abelian.
\end{theorem}

\begin{remark}
  Here, the field of coefficients of the NVE is the field of
meromorphic functions on $\gamma$.
\end{remark}

\section{Method to determine families of Hamiltonians with specific NVE}
This method was implemented in \cite{acbl} as a generalization of
the method shown in \cite{mosi}. This section is devoted to this
method.

\medskip

  Let us consider a classical Hamiltonian of two degrees of freedom,
$$H = \frac{y_1^2+y_2^2}{2} + V(x_1,x_2).$$
  $V$ is the \emph{potential function}, and it is assumed to be
analytical in some open subset of $\mathbb{C}^2$. The evolution of
the system is determined by Hamilton equations:

$$ \dot x_1 = y_1,\quad \dot x_2 = y_2, \quad \dot y_1 =
-\frac{\partial V}{\partial x_1},\quad \dot y_2 = -\frac{\partial
V}{\partial x_2}.$$

  Let us assume that the plane $\Gamma = \{x_2=0, y_2 = 0\}$ is
an invariant manifold of the Hamiltonian. We keep in mind that the
family of integral curves lying on $\Gamma$ is parameterized by
the energy $h = H|_\Gamma$, but we do not need to use it
explicitly. We are interested in studying the linear approximation
of the system near $\Gamma$. Since $\Gamma$ is an invariant
manifold, we have
$$\left.\frac{\partial V}{\partial x_2}\right|_\Gamma = 0,$$
so that the NVE for a particular solution
$$t\mapsto\gamma(t) =(x_1(t), y_1 = \dot x_1(t), x_2 = 0, y_2 = 0),$$
is
$$\dot \xi = \eta,\quad \dot \eta = -\left[\frac{\partial^2 V}{\partial x_2^2}(x_1(t),0)\right] \xi.$$

Let us define
$$\phi(x_1) = V(x_1, 0),\quad \alpha(x_1) =
- \frac{\partial^2 V}{\partial x_2^2}(x_1,0),$$ and then we write
the second order Taylor series in $x_2$ for $V$, obtaining the
following expression for $H$
\begin{equation}\label{:Hgeneral}
  H = \frac{y_1^2+y_2^2}{2} +  \phi(x_1) - \alpha(x_1)\frac{x_2^2}{2}
  + \beta(x_1,x_2)x_2^3,
\end{equation}
which is the \emph{general form of a classical analytic
Hamiltonian, with invariant plane $\Gamma$}, provided that a
Taylor expansion of the potential around $\{x_2=0\}$ exists. The
NVE associated to any integral curve lying on $\Gamma$ is
\begin{equation}\label{:NVE}
\ddot \xi = \alpha(x_1(t)) \xi.
\end{equation}

\subsection{General Method}

  We are interested in computing Hamiltonians of the family
(\ref{:Hgeneral}), such that its NVE (\ref{:NVE}) belongs to a
specific family of Linear Differential Equations (LDE). Then we
can apply our results about the integrability of this LDE, and the
Morales-Ramis theorem to obtain information about the
non-integrability of such Hamiltonians.
\\

  From now on, we will write $a(t) = \alpha(x_1(t))$ for a generic
curve $\gamma$ lying on $\Gamma$, parameterized by $t$. Then, the
NVE is
\begin{equation}\label{:NVEa}
\ddot \xi = a(t)\xi.
\end{equation}

The following step is to consider a differential polynomial
$Q(\eta,\dot\eta,\ddot\eta,\ldots)
\in\mathbb{C}[\eta,\dot\eta,\ddot\eta,\ldots]$, being $\eta$ a
differential indeterminate ($Q$ is polynomial in $\eta$ and a
finite number of the successive derivatives of $\eta$). After, we
need to compute all Hamiltonians in the family \eqref{:Hgeneral}
such that for any particular solution in $\Gamma$, the coefficient
$a(t)$ of the corresponding NVE satisfies $Q(a, \dot a, \ddot a,
\ldots ) = 0$.

\medskip

  We should notice that for a generic integral curve
$\gamma(t) = (x_1(t),y_1 = \dot x_1(t))$ lying on $\Gamma$,
(\ref{:NVEa}) depends only of the values of the functions $\alpha$
and $\phi$. It depends on $\alpha(x_1)$, since $a(t) =
\alpha(x_1(t))$. We observe that the curve $\gamma(t)$ is a
solution of the restricted Hamiltonian
\begin{equation}\label{:restrictedH}
  h = \frac{y_1^2}{2} + \phi(x_1)
\end{equation}
  whose associated Hamiltonian vector field is
\begin{equation}\label{:Xh}
  X_h = y_1\frac{\partial}{\partial x_1} - \frac{d\phi}{d x_1}\frac{\partial}{\partial
  y_1}.
\end{equation}
 Thus $x_1(t)$ is a solution of the differential equation $\ddot x_1 = -\frac{d\phi}{d x_1},$
and then, the relation of $x_1(t)$ is given by $\phi$.

\medskip

 Since $\gamma(t)$ is an integral curve of $X_h$, for any function
$f(x_1,y_1)$ defined in $\Gamma$ we have
  $$\frac{d}{d t}\gamma^*(f) = \gamma^*(X_h f),$$
where $\gamma^*$ denotes the usual pullback of functions. Then,
using $a(t) = \gamma^*(\alpha)$, we have for each $k\geq 0$,
\begin{equation}
\frac{d^k a}{d t^k} = \gamma^*(X_h^k\alpha),
\end{equation}
so that
$$Q(a,\dot a, \ddot a,\ldots ) = Q(\gamma^*(\alpha),\gamma^*(X_h\alpha),\gamma^*(X_h^2\alpha),\ldots).$$

  There is an integral curve of the Hamiltonian passing through each point of $\Gamma$, so that
we have proven the following.

\begin{proposition}\label{:prop2}
  Let $H$ be a Hamiltonian of the family {\rm (\ref{:Hgeneral})}, and let $Q(a,\dot a,\ddot a,\ldots)$ be a differential
polynomial with constants coefficients. Then, for each integral
curve lying on $\Gamma$, the coefficient $a(t)$ of the NVE {\rm
(\ref{:NVEa})} verifies $Q(a,\dot a, \ddot a, \ldots,) = 0$ if and
only if the function
$$ \hat Q(x_1,y_1) = Q(\alpha, X_h\alpha, X_h^2 \alpha, \ldots)$$
vanishes on $\Gamma$.
\end{proposition}

\begin{remark}
  In fact, the NVE of an integral curve depends on the
  parameterization, while our criterion does not depend on any choice of parameterization of
the integral curves. We observe that a polynomial $Q(a,\dot a,
\ddot a,\ldots)$ with constant coefficients is an invariant of the
group by translations of time.
\end{remark}

  Next, we will see that $\hat Q(x_1,y_1)$ is a polynomial in $y_1$ and its coefficients are
differential polynomials in $\alpha, \phi$. If we write down the
expressions for successive Lie derivatives of $\alpha$, we obtain
\begin{align*}\begin{split}
X_h\alpha&=y_1\frac{d\alpha}{dx_1},\\
X_h^2\alpha&=y_1^2\frac{d^2\alpha}{dx_1^2}-\frac{d\phi}{dx_1}\frac{d\alpha}{dx_i},\\
X_h^3\alpha&=y_1^3\frac{d^3\alpha}{dx_1^3}-y_1\left(\frac{d}{dx_1}\left(\frac{d\phi}{dx_1}\frac{d\alpha}{dx_1}\right)
    +2\frac{d\phi}{dx_1}\frac{d^2\alpha}{dx_1^2}\right)
\end{split}\end{align*}

  In general form we have
\begin{equation}\label{:induction}
X_h^{n+1}\alpha = y_1\frac{\partial X_h^{n}\alpha}{d x_1} -
\frac{d\phi}{d x_1}\frac{\partial X_h^{n}}{\partial y_1};
\end{equation}
it inductively follows that they all are polynomials in $y_1$, in
which their coefficients are differential polynomials in $\alpha$
and $\phi$. If we write $X_h^n\alpha$ explicitly,
\begin{equation}\label{:Ex}
X_h^n\alpha = \sum_{n\geq k \geq 0} E_{n,k}(\alpha,\phi)y_1^k,
\end{equation}
we can see that the coefficients $E_{n,k}(\alpha,\phi) \in
\mathbb{C}\left[\alpha,\phi,\frac{d^r\alpha}{d
x_1^r},\frac{d^s\phi}{dx_1^s}\right]$ satisfy the following
recurrence law
\begin{equation}\label{:rlaw}
  E_{n+1,k}(\alpha,\phi) = \frac{d}{d x_1} E_{n,k-1}(\alpha,\phi) - (k+1)E_{n,k+1}(\alpha,\phi)\frac{d\phi}{d x_1}
\end{equation}
with initial conditions
\begin{equation}\label{:condition}
E_{1,1}(\alpha,\phi) = \frac{d \alpha}{d x_1},\quad
E_{1,k}(\alpha,\phi) = 0 \,\,\,\forall k\neq 1.
\end{equation}

\begin{remark}\label{:remark}
  The recurrence law (\ref{:rlaw}) and the initial conditions (\ref{:condition}) determine the coefficients
$E_{n,k}(\alpha,\phi)$. We can compute the values of some of them
easily:
\begin{itemize}
\item $E_{n,n}(\alpha,\phi) = \frac{d^n\alpha}{d x_1^n}$ for all $n \geq 1$.
\item $E_{n,k}(\alpha,\phi)= 0$ if $n-k$ is odd, or $k<0$, or $k>n$.
\end{itemize}
\end{remark}

\section{Main Result: Families of Hamiltonian systems with quartic NVEs}

\begin{theorem}
  Let $H = T + V$ be a classical Hamiltonian with invariant plane $\Gamma$
such that the generic NVE along integral curves in $\Gamma$ is a
Hill-Schr\"odinger equation with quartic polynomial coefficient.
Then, the potential $V$ up to constants corresponds to the
potential of the generalized Hen\'on-Heiles system \eqref{genhh}
with $n=4$ and $\lambda=0$.
\end{theorem}

\begin{proof}
  Following our general method, the family of potentials
satisfying the assumptions of the theorem are given by the
solutions $\alpha(x_1)$, $\phi(x_1)$ of the system of differential
equations
$$E_{5,5}(\alpha,\phi) = 0,\quad  E_{5,3}(\alpha,\phi) = 0,\quad
E_{5,1}(\alpha,\phi)=0.$$ The first equation is just the following
$$E_{5,5}(\alpha,\phi) = \frac{d^5\alpha}{d x_1^5}=0,$$
and then we know that $\alpha$ is a quartic polynomial in $x_1$,
$$\alpha = a + b x_1 + c x_1^2 + d x_1^3 + e x_1^4,$$
where $a,b, c, d, e$ are complex numbers and $e$ does not vanish.
Then we substitute $\alpha$ into the equations
$$E_{5,3}(\alpha,\phi)=0,\quad E_{5,1}(\alpha,\phi) = 0,$$
obtaining in this way the following system of differential
equations in $\phi(x_1)$, being $'=d/dx$:
$$(4ex_1^3 + 3dx_1^2 + 2cx_1 + b)\phi^{iv} +
(60ex_1^2+30dx_1+10c)\phi''' + $$
\begin{equation}\label{:L}
+(240ex_1+60d)\phi''+240 e\phi'=0, \tag{L}
\end{equation}

\medskip

$$(18d+72ex_1)(\phi')^2 + (b+2cx_1+3dx_1^2+4ex_1^3)(\phi'')^2 +$$
\begin{equation}\label{:NL}
+(14c + 42dx_1 +
84ex_1^2)\phi'\phi''+(b+2cx_1+3dx_1^2+4ex_1^3)\phi'\phi'''=0.
\tag{NL}
\end{equation}

Equations \eqref{:L} and \eqref{:NL} are ordinary differential
equations in $\phi'$ with some complex parameters. We can
substitute a new unknown $y$ for $\phi'$ in order to reduce the
order by one. Secondly, by a translation of $x_1$ by a scalar
value
$$x = x_1 - \mu,$$
we can assume that one of the coefficients of the polynomial
$\alpha(x_1)$ vanishes. From now on let us write
$$\alpha(x) = a + b x + c x^2 + e x^4,$$
and let us study the system of differential equations:

\begin{equation}\label{:L2}
(4ex^3 + 2cx + b)y''' + (60ex^2+10c)y''+ 240exy'+ 240ey=0,
\tag{L2}
\end{equation}
\begin{equation}\label{:NL2}
72exy^2 + (b+2cx+4ex^3)(y')^2+(14c+84ex^2)yy'+(b+2cx+4ex^3)yy''=0.
\tag{NL2}
\end{equation}

The first equation \eqref{:L2} is a linear equation in $y$. In
this special case we will be able to completely solve the equation
$\eqref{:L2}$ and then prove that solutions of \eqref{:L2} do not
satisfy the non-linear equation \eqref{:NL2}. Then, the only
solution of the system is given by the function $y=0$ that
corresponds to $\phi = \lambda_0\in\mathbb{C}$, and then the
potential
$$V = \phi + \frac{\alpha(x_1)}{2}x_2^2 + \beta(x_1,x_2)x_2^3,$$
is of the form given in the statement of the theorem.

\medskip

{\bf Solution of the equation \eqref{:L2} }

This linear equation is solvable by elementary methods.
Fortunately, its Galois group is trivial, and therefore we can
look for a fundamental system of solutions that are rational
functions over $x,a,b,c,e$. The main problem is that such a system
of solutions does not always specialize to a particular system of
solutions when fixing the values of the parameters $a,b,c,e$.
There are some values of the parameters that correspond to
degeneracy of the system of solutions. The equations of this
\emph{locus of degeneration} are given by the wronskian of the
fundamental system. When the wronskian vanishes, the fundamental
system degenerates, and then a different solution appears. We have
to consider also these restricted problems independently.

\medskip

First, we find the general solution of \eqref{:L2} for generic
values of the parameters, depending on arbitrary constants $K_1$,
$K_2$, $K_3$:
\begin{equation}\label{:SolL2}
y = \frac{K_1 N_1 + K_2 N_2 + K_3 N_3}{D^3} = \frac{P}{D^3},
\end{equation}
where
$$D = 4ex^3+2cx+b,$$
$$N_1 = x(4ec^2x^5-42becx^4-(6c^3+48eb^2)x^3+9b^2cx+6b^3),$$
$$N_2 = x(8ecx^5-12bcex^4-(24eb^2+12c^3)x^3-12bc^2x^2+3b^3),$$
$$N_3 = 8c^2e^2x^6 -84bce^2x^5 -
(12c^3e+168b^2e^2)x^4+21b^3ex-3b^2c^2.$$

Let us study for which parameters the above expression is not the
general solution of \eqref{:L2}. This happens if and only if the
wronskian of the fundamental system of solutions vanishes. We know
that the wronskian of the fundamental solutions $N_i/D$ vanishes
if and only if the wronskian of the numerators $N_i$ vanishes. We
compute it, obtaining

$$\begin{array}{lll}W(N_1,N_2,N_3)&=& 162c^3b^7 + 1296b^6c^4x +
3888b^5c^5x^2 \\& &+ (2592b^6ec^3+5184b^4c^6)x^3
+(2592c^7b^3+15552b^5c^4e)x^4 \\ & & +31104b^4c^5ex^5  +
(15552b^5c^3e^2+20736b^3c^6e)x^6 \\ & & + 62208b^4c^4e^2 x^7
 + 62208b^3c^5e^2x^8 + 41472b^4c^3e^3x^9 \\ & & + 82944
b^3c^4e^3x^{10} + 41472 b^3c^3e^4 x^{12}. \end{array}$$ The
equation $W(N_1,N_2,N_3)=0$ is simple to solve and it has two
independent solutions, that we will consider independently.
\\

\begin{enumerate}
\item[(a)] \{$b = 0$\}
\item[(b)] \{$c = 0$\}
\end{enumerate}

\medskip

{\bf Case A, $b=0$.}

If $b$ vanishes then the system of equations is:
\begin{equation}\label{:L3}
(4ex^3 + 2cx)y''' + (60ex^2+10c)y''+ 240exy'+ 240ey=0, \tag{L3}
\end{equation}
\begin{equation}\label{:NL3}
72exy^2 + (2cx+4ex^3)(y')^2+(14c+84ex^2)yy'+(2cx+4ex^3)yy''=0.
\tag{NL3}
\end{equation}
We obtain a new general solution for this restricted case by
direct integration of the linear equation. The general solution is:
\begin{equation}\label{:solL3}
y = \frac{K_1N_{31}+K_2N_{32}+K_3N_{33}}{D_3^3} = \frac{P_3}{D_3^3},
\end{equation}
where
$$D_3=2ex^2+c,$$
$$N_{31}=6ex^2-c,$$
$$N_{32}=x(-3c+2ex^2),$$
$$N_{33}=(c^3+6ec^2x^2+16e^3x^6)x^{-3}.$$

In this case, the wronskian of the numerators is,
$$W(N_{31},N_{32},N_{33}) = \frac{96ec(6ec^2x^4+16e^3x^6+5c^3)}{x^4},$$
so that this system of solutions is degenerated only when $c=0$,
which is considered as a particular case of the following.

\medskip

{\bf Case B, $c = 0$.}

Let us finally consider the last case $c=0$. The system of
equations is now as follows.
\begin{equation}\label{:L4}
(4ex^3 + b)y''' + 60ex^2y''+ 240exy'+ 240ey=0, \tag{L4}
\end{equation}
\begin{equation}\label{:NL4}
72exy^2 + (b+4ex^3)(y')^2+84ex^2yy'+(b+4ex^3)yy''=0. \tag{NL4}
\end{equation}
The general solution of \eqref{:L4} is given by,
\begin{equation}\label{:solL4}
y=\frac{K_1N_{41}+K_2N_{42}+K_3N_{43}}{D_4^3}=\frac{P^3}{D_4^3}
\end{equation} where
$$D_4 = 4ex^3+b,$$
$$N_{41}=x(b-8ex^3),$$
$$N_{42}=x^2(b-2ex^3)$$
$$N_{43}=b^2-28ebx^3+16e^2x^6.$$
We analyze when the wronskian
$$W(N_{41},N_{42},N_{43}) = 2b^4 + 32eb^3x^3 + 192e^2b^3x^6 +
512be^3x^9 + 512e^4x^{12}$$ vanishes. Because of the coefficient
in the 12th power of $x$, this wronskian does not vanish for any
of the considered values of the parameters, and then this general
solution does not degenerate.

\medskip

{\bf Common solutions with the non-linear equation}

Here we look for solutions of the linear equation that also
satisfy the considered non-linear equation. We directly substitute
the general solution \eqref{:SolL2} of \eqref{:L2} into
\eqref{:NL2}. Then we obtain a rational expression:
$$\frac{Q(x;b,c,e,K_1,K_2,K_3)}{D^7} = 0,$$
with $Q$ a polynomial in $x$ depending on the parameter $b,c,e,
K_1, K_2, K_3$. Thus, we look for the values of the parameters
that force $Q$ to vanish. If we develop $Q$ as a differential
polynomial in $P$ and $D$ we obtain the following expression: $Q =
(72exD+((-6 + 42c + 262ex^3)DD'+21D'^2+D^2-3DD'')P^2
+(14c+84ex)(D^2-6DD')PP'+D^2P^2$. Note that the polynomials $D,
D', D''$ do not depend on the parameters $K_i$, and $P, P'$ are
linear in such parameters. It follows that $Q$ is polynomial in
$x$ of degree $16$ whose coefficients are homogeneous polynomials
of degree two in the parameters $K_1, K_2, K_3$, that is,
$$Q = \sum_{i=0}^{16} C_i(K_1,K_2,K_3,b,c,e) x^i = 0,$$
where:
$$C_i = \left(K_1, K_2, K_3\right) \left(\begin{matrix} \lambda^i_{11} & \lambda^i_{12} &
\lambda^i_{13} \\ \lambda^i_{21} & \lambda^i_{22} & \lambda^i_{23} \\
\lambda^i_{31} & \lambda^i_{32} & \lambda^i_{33}
\end{matrix}\right) \left(\begin{matrix}K_1 \\ K_2 \\ K_3
\end{matrix}\right).$$ Finally, the coefficients $\lambda^i_{jk}$
are polynomials in the parameters $b,c,e$. The  common solution of
\eqref{:L2} and \eqref{:NL2} corresponds to values of the
parameters $b,c,e,K_1,K_2,K_3$ that are solutions of the system of
17 algebraic equations:
$$C_i(K_1,K_2,K_3,b,c,e) = 0, \quad\quad i = 0,\ldots,16.$$

Each equation $C_i=0$ is the equation of a $3$-dimensional cone in
the affine space over the field $\mathbb C(b,c,e)$; thus the
equation of a conic curve in the projective plane $\mathbb
P(\mathbb C(b,c,e))$ of homogeneous coordinates $K_1, K_2, K_3$.
Two conic curves intersect in four points. This simplifies the
computations, since we considerate some proper subset of 17 equations and check the incompatibility of it, provided that $b$ and $c$ are
different from zero.

Following the same schema, we analyze the exceptional case $b=0$.
We put the solution \eqref{:solL3} of \eqref{:L3} into the
equation \eqref{:NL3}. We obtain an expression
$$ \frac{Q_3}{x^7D_3^7}= 0,$$
where $Q_3 = (72exD_3 + ((-6 + 42c + 262ex^3)D_3D'_3 + 21D_3'^2 +
D_3^2 - 3D_3D_3'')P_3^2 +(14c+84ex)(D_3^2-6D_3D'_3)P_3P_3' +
D_3^2P_3^2$. Here, we find that $Q_3$ is a polynomial in $x$ of
degree $18$,
$$Q_3 = \sum_{i=0}^{18} E_i(c,e,K_1,K_2,K_3) x^i,$$
and again the system of algebraic equations $\{E_i = 0\}$ is
qualitatively similar to the above system $\{C_i=0\}$. The same
analysis is carried out in the other exceptional case $c=0$. In
this last case we obtain that $Q_4 = (72exD_4 + (-6 +
262ex^3)D_4D'_4 + 21D_4'^2 + D_4^2 - 3D_4D_4'')P_4^2
+84exD_4^2-6D_4D'_4)P_4P_4' + D_4^2P_4^2$ is of degree $18$ in
$x$, and then we have a system $\{F_i=0\}$ of $19$ algebraic
equations that form an incompatible system.
\end{proof}

\section{Final comments and open questions}
One open problem presented in \cite{acbl} is the problem of
determining families of classical Hamiltonians with an invariant
plane and NVE of Hill-Schr\"odinger type whose polynomial
coefficient is of even degree greater than two. Currently the case
of even degree greater than four is still open.
\\

It is very well known that the classical Hen\'on-Heiles system
with $A=0$ and $\lambda=6$ is integrable. One question is as
follows: how must be $\lambda$, $\beta(x_1,x_2)$ and the
coefficients $A_0,\ldots,A_n$ to obtain integrable general
Hen\'on-Heiles systems?
\\

The problem of analyzing the monodromy of the NVE of integral
curves of a two degrees of freedom Hamiltonian (both classical and
general) was studied by Baider, Churchill and Rod at the beginning
of the 90's (see \cite{bachro}). Their method is quite different,
they imposed the monodromy group to verify some special properties
that were translated as algebraic conditions in the Hamiltonian
functions. Their theory was restricted to the case of fuchsian
groups, which in terms of Galois theory means regular
singularities, while we work in the general case. It would be very
interesting to compare both approaches.
\\

 Another problem is to apply these methods
to higher variational equations and non-autonomous hamiltonian
systems.

\section*{Acknowledgements}
The research of the first author is partially supported by grant
FPI Spanish Government, project DGICYT MTM 2006-00478. The second
and third authors acknowledge the Universidad Sergio Arboleda,
specially Reinaldo N\'u\~nez and Jes\'us Hernando P\'erez for
their motivation and support. Finally, we would like to thank Anna
de Mier for their valuable comments and suggestions.

\end{document}